\documentclass[11pt]{article}%
\usepackage{amsfonts}
\usepackage{amsmath}
\usepackage{amssymb}
\usepackage{graphicx}%
\providecommand{\U}[1]{\protect\rule{.1in}{.1in}}
\newtheorem{theorem}{Theorem}

\newenvironment{proof}[1][Proof]{\noindent\textbf{#1.} }{\ \rule{0.5em}{0.5em}}
\begin{document}

\title{On the Disentanglement of Gaussian Quantum States by Symplectic Rotations\\Sur la D\'{e}sintrication des \'{E}tats Quantiques Gaussiens par des Rotations Symplectiques}
\author{Maurice A. de Gosson\thanks{maurice.de.gosson@univie.ac.at}\\Universit\"{a}t Wien\\Fakult\"{a}t f\"{u}r Mathematik (NuHAG)\\Oskar-Morgenstern-Platz 1, 1090 Wien (AUSTRIA)}
\maketitle

\begin{abstract}
We show that every Gaussian mixed quantum state can be disentangled by
conjugation with a unitary operator corresponding to a symplectic rotation via
the metaplectic representation of the symplectic group. The main tools we use
are the Werner--Wolf condition for separability on covariance matrices and the
symplectic covariance of Weyl pseudo-differential operators.

\end{abstract}
\begin{abstract}
Nous montrons que chaque \'{e}tat quantique Gaussien peut-\^{e}tre rendu
s\'{e}parable (= \textquotedblleft d\'{e}sintriqu\'{e}\textquotedblright) par
conjugaison avec un op\'{e}rateur unitaire associ\'{e} via le groupe
m\'{e}taplectique \`{a} une rotation symplectique. Pour cela nous utilsons la
condition de s\'{e}parabilit\'{e} de Werner et Wolf sur la matrice de
covariance ainsi que la covariance symplectique des op\'{e}rateurs
pseudo-diff\'{e}rentiels de Weyl.

\end{abstract}

\section{Introduction}

Gaussian states play an ubiquitous role in quantum information theory and in
quantum optics because they are easy to manufacture in the laboratory, and
have in addition important extremality properties \cite{wolf}. Of particular
interest are the separability and entanglement properties of Gaussian states;
the literature on the topic is immense; two excellent texts whose mathematical
setup is rigorous are \cite{AI,aragy}. It turns out that even if major
advances have been made in the study of the separability of Gaussian quantum
states in recent years (one of the milestones being Werner and Wolf's paper
\cite{ww} about the covariance matrices of bipartite states), the topic is
still largely open. The aim of this Note is to show that every Gaussian state
can be made separable by using a symplectic rotation and of the corresponding
metaplectic operator. (We note that physicists use the terminology
\textquotedblleft passive symplectic transformations\textquotedblright\ in
place of \textquotedblleft symplectic rotation\textquotedblright). This result
can be viewed as closing a problem originally posed in Wolf \textit{et al}.
\cite{Wolf}, who asked which Gaussian states can be entangled by symplectic
rotations. A full answer has recently been given in \cite{Lami} \textit{et
al}. where the Gaussian states that are separable for all symplectic rotations
are characterized. Our result (Theorem \ref{Thm1}) shows that, conversely,
every entangled Gaussian state can be separated (\textquotedblleft
disentangled\textquotedblright) by metaplectic transformations corresponding
to symplectic rotations.

We will use the following notation. Let $\mathbb{R}^{2n}=\mathbb{R}^{2n_{A}%
}\oplus\mathbb{R}^{2n_{B}}$ be the phase space of a bipartite system
($n_{A}\geq1$, $n_{B}\geq1$). We will use the following phase space variable
ordering: $z=(z_{A},z_{B})=z_{A}\oplus z_{B}$ with $z_{A}=(x_{1}%
,p_{1},...,x_{n_{A}},p_{n_{A}})$ and $z_{B}=(x_{n_{A}+1},p_{n_{A}+1},$
$...,x_{n},p_{n})$. We equip the symplectic spaces $\mathbb{R}^{2n_{A}}$ and
$\mathbb{R}^{2n_{B}}$ with their canonical bases. The symplectic structure on
$\mathbb{R}^{2n}$ is then $\sigma(z,z^{\prime})=Jz\cdot z^{\prime}$ with
$J=J_{A}\oplus J_{B}$ where
\[
J_{A}=\bigoplus_{k=1}^{n_{A}}J_{k}\text{ \ , \ }J_{k}=%
\begin{pmatrix}
0 & 1\\
-1 & 0
\end{pmatrix}
\]
and likewise for $J_{B}$. Thus $J_{A}$ (\textit{resp.} $J_{B}$) determines the
symplectic structure on the partial phase space $\mathbb{R}^{2n_{A}}$
(\textit{resp}. $\mathbb{R}^{2n_{B}}$).

\section{Result: Statement and Proof}

Let $\Sigma$ be a real positive definite symmetric $2n\times2n$ matrix (to be
called \textquotedblleft covariance matrix\textquotedblright\ from now on) and
consider the associated normal probability distribution
\begin{equation}
\rho(z)=\frac{1}{(2\pi)^{n}\sqrt{\det\Sigma}}e^{-\frac{1}{2}\Sigma^{-1}z^{2}%
}~. \label{rho}%
\end{equation}
If the covariance matrix satisfies in addition the condition
\begin{equation}
\Sigma+\frac{i\hbar}{2}J\geq0 \label{quant1}%
\end{equation}
($J$ the standard symplectic matrix) then $\rho$ is the Wigner distribution of
a mixed quantum state, identified with its density operator $\widehat{\rho}$.
We notice that property (\ref{quant1}) crucially depends on the numerical
value of $\hbar$ (see \cite{DiPr09,go17}). We will say that $\widehat{\rho}$
is \textquotedblleft$AB$-separable\textquotedblright\ if there exist sequences
of density operators $(\widehat{\rho}_{j}^{A})$ and $(\widehat{\rho}_{j}^{B})$
on $L^{2}(\mathbb{R}^{n_{A}})$ and $L^{2}(\mathbb{R}^{n_{B}})$, respectively
and coefficients $\lambda_{j}\geq0$ summing up to one, such that
\begin{equation}
\widehat{\rho}=\sum_{j}\lambda_{j}\widehat{\rho}_{j}^{A}\otimes\widehat{\rho
}_{j}^{B} \label{AB}%
\end{equation}
where the convergence is for the trace-class norm. The problem of determining
necessary and sufficient conditions for a density operator to be separable is
still very largely open; while there exist necessary conditions, no simple
sufficient condition for separability is known in the general case; for a
recent up to date discussion see Lami \textit{et al}. \cite{Lami}. Werner and
Wolf \cite{ww} have proven that in the Gaussian case $\widehat{\rho}$ is
separable if and only if there exists a $2n_{A}\times2n_{A}$ covariance matrix
$\Sigma_{A}$ and a $2n_{B}\times2n_{B}$ covariance matrix $\Sigma_{B}$ such
that the following conditions hold:
\begin{gather}
\Sigma_{A}+\frac{i\hbar}{2}J_{A}\geq0\text{ }\label{ww1}\\
\text{\ }\Sigma_{B}+\frac{i\hbar}{2}J_{B}\geq0\text{\ }\label{ww1,5}\\
\Sigma\geq\Sigma_{A}\oplus\Sigma_{B}~. \label{ww2}%
\end{gather}

The aim of this Letter is to prove that for every Gaussian density operator
there exists a unitary transform $\widehat{U}$ such that $\widehat{U}%
\widehat{\rho}\widehat{U}^{-1}$ is a separable Gaussian state:

\begin{theorem}
\label{Thm1}Let $\widehat{\rho}$ be a density operator with Gaussian Wigner
distribution (\ref{rho}). There exists a symplectic rotation $U\in U(n)$
($=\operatorname*{Sp}(n)\cap O(2n,\mathbb{R})$) such that $\widehat{U}%
\widehat{\rho}\widehat{U}^{-1}$ is separable where $\widehat{U}\in
\operatorname*{Mp}(n)$ is any of the two metaplectic operators covering $U$.
\end{theorem}

\begin{proof}
We begin by recalling \cite{FP,Birk} that the quantum condition (\ref{quant1})
is equivalent to the statement:%
\begin{equation}
\text{\emph{There exists} }S\in\operatorname*{Sp}(n)\text{ \emph{such that}
}SB^{2n}(\sqrt{\hbar})\subset\Omega_{\Sigma} \label{quant2}%
\end{equation}
where $\operatorname*{Sp}(n)$ is the symplectic group of the phase space
$\mathbb{R}^{2n}\equiv\mathbb{R}_{x}^{n}\times\mathbb{R}_{p}^{n}$ equipped
with the standard symplectic form
\[
\sigma=dp_{1}\wedge dx_{1}+\cdot\cdot\cdot+dp_{n}\wedge dx_{n}~,
\]
$B^{2n}(\sqrt{\hbar})$ is the phase space ball defined by $|z|\leq\hbar$ and
$\Omega_{\Sigma}$ the covariance ellipsoid of $\widehat{\rho}$:%
\[
\Omega_{\Sigma}=\{z\in\mathbb{R}^{2n}:\tfrac{1}{2}\Sigma^{-1}z^{2}\leq1\}~.
\]
Let $S=PR$ ($P=(S^{T}S)^{1/2}$, $R=(S^{T}S)^{-1/2}S$) be the symplectic polar
decomposition \cite{Birk} of $S\in\operatorname*{Sp}(n)$, that is
$P\in\operatorname*{Sp}(n)$, $P>0$, and
\[
R\in U(n)=\operatorname*{Sp}(n)\cap O(2n,\mathbb{R})~.
\]
We have $SB^{2n}(\sqrt{\hbar})=PB^{2n}(\sqrt{\hbar})$ by rotational symmetry
of the ball $B^{2n}(\sqrt{\hbar})$. There exists a symplectic rotation $U\in
U(n)$ diagonalizing $P$ \cite{Birk}:%
\begin{equation}
P=U^{T}\Delta U \label{factor}%
\end{equation}
where $\Delta\in\operatorname*{Sp}(n)$ is a diagonal matrix whose form will be
discussed in a moment. The inclusion $SB^{2n}(\sqrt{\hbar})\subset
\Omega_{\Sigma}$ in (\ref{quant2}) is thus equivalent to $\Delta B^{2n}%
(\sqrt{\hbar})\subset U(\Omega_{\Sigma})$, that is
\begin{equation}
\Delta B^{2n}(\sqrt{\hbar})\subset\Omega_{\Sigma_{U}} \label{quant3}%
\end{equation}
where $\Sigma_{U}=U\Sigma U^{T}$. This inclusion is equivalent to the matrix
inequality%
\begin{equation}
\frac{\hbar}{2}\Delta^{2}\leq\Sigma_{U} \label{quant4}%
\end{equation}
($A\leq B$ meaning that $B-A$ is positive semidefinite). We next note that
$\Sigma_{U}$ is the covariance matrix of the density operator $\widehat{\rho
}_{U}$ with Wigner distribution $\rho_{U}(z)=\rho(U^{T}z)$ that is
\[
\rho_{U}(z)=\frac{1}{(2\pi)^{n}\sqrt{\det U\Sigma U^{T}}}e^{-\frac{1}{2}%
\Sigma^{-1}U^{T}z\cdot U^{T}z}~.
\]
Recall now the following symplectic covariance property: if $\widehat{A}%
=\mathrm{Op}^{\mathrm{W}}(a)$ is a Weyl operator with symbol $a$ and
$\widehat{S}\in\operatorname*{Mp}(n)$ a metaplectic operator covering
$S\in\operatorname*{Sp}(n)$ then
\begin{equation}
\widehat{S}\mathrm{Op}^{\mathrm{W}}(a)\widehat{S}^{-1}=\mathrm{Op}%
^{\mathrm{W}}(a\circ S^{-1}) \label{sympco}%
\end{equation}
(see for instance \cite{Littlejohn} or \cite{Birk}, Ch.7). Applying this
covariance formula to $\widehat{\rho}=(2\pi\hbar)^{n}\mathrm{Op}^{\mathrm{W}%
}(\rho)$ yields since $U^{T}=U^{-1}$,
\begin{equation}
\widehat{\rho}_{U}=\widehat{U}\widehat{\rho}\widehat{U}^{-1} \label{RU}%
\end{equation}
where $\widehat{U}$ is anyone of the two metaplectic operators $\pm
\widehat{U}$ covering $U$. We claim that $\widehat{\rho}_{U}$ is separable. To
see this, let us come back to the diagonal matrix $\Delta$ appearing in the
factorization $P=U^{T}\Delta U$ (\ref{factor}). Its diagonal elements are the
eigenvalues $\lambda_{1},...,\lambda_{2n}$ of the positive definite symplectic
matrix $P$ and therefore appear in pairs $(\lambda,1/\lambda)$ with
$\lambda>0$ \cite{Arnold,Birk}. In fact, in the $AB$-ordering we are using,
the matrix $\Delta$ has the form $\Delta=\Delta_{A}\oplus\Delta_{B}$ with
\[
\Delta_{A}=\bigoplus_{k=1}^{n_{A}}\Delta_{k}\text{ \ , \ }\Delta_{B}%
=\bigoplus_{k=n_{A}+1}^{n}\Delta_{k}%
\]
and $\Delta_{k}=%
\begin{pmatrix}
\lambda_{k} & 0\\
0 & \lambda_{k}^{-1}%
\end{pmatrix}
$ for $k=1,...,n$. Clearly $\Delta_{A}\in\operatorname*{Sp}(n_{A})$ and
$\Delta_{B}\in\operatorname*{Sp}(n_{B})$. The symmetric matrices
\[
\Sigma_{A}=\frac{\hbar}{2}\Delta_{A}^{2}\text{ \ },\text{ \ }\Sigma_{B}%
=\frac{\hbar}{2}\Delta_{B}^{2}%
\]
trivially satisfy $\Sigma_{A}+\frac{i\hbar}{2}J_{A}\geq0$ and $\Sigma
_{B}+\frac{i\hbar}{2}J_{B}\geq0$. In view of (\ref{quant4}) we have
\[
\Sigma_{A}\oplus\Sigma_{B}\leq\Sigma_{U}%
\]
and the theorem now follows using the Werner--Wolf conditions (\ref{ww1}%
)---(\ref{ww2}).
\end{proof}

\end{document}